\documentclass{article}
\usepackage{arxiv}
\usepackage[utf8]{inputenc} 
\usepackage[T1]{fontenc}    
\usepackage{hyperref}       
\usepackage{url}            
\usepackage{booktabs}       
\usepackage{amsfonts}       
\usepackage{nicefrac}       
\usepackage{microtype}      
\usepackage{custom}
\usepackage{graphicx}
\usepackage{float}
\usepackage{appendix}
\usepackage{subcaption}
\usepackage[dvipsnames]{xcolor}
\usepackage[algo2e,ruled,vlined]{algorithm2e}
\usepackage{IEEEtrantools}
\usepackage{amsmath}
\usepackage{array}
\usepackage{multirow}
\usepackage{color}
\usepackage{xr}
\usepackage{enumitem}
\usepackage{bm}
\usepackage[mathscr]{euscript}

\newtheorem{thm}{\bf{Theorem}}

\newtheorem{lem}{\bf{Lemma}}
\theoremstyle{definition}
\newtheorem{defn}{Definition}
\theoremstyle{remark}

\title{Model-Aware Regularization for Deep Learning Approaches to Inverse Problems}

\author{%
  Jaweria~Amjad \\
  Department of E \& EE\\
  University College London\\
  London, United Kingdom \\
  \texttt{jaweria.amjad.16@ucl.ac.uk}
   \And
  Zhaoyan Lyu \\
  Department of E \& EE\\
  University College London\\
  London, United Kingdom \\
  \texttt{z.lyu.17@ucl.ac.uk}
   \And
   Miguel R. D.~Rodrigues \\
  Department of E \& EE\\
  University College London\\
  London, United Kingdom \\
  \texttt{m.rodrigues@ucl.ac.uk}
}

\begin{document}

\maketitle

\begin{abstract}
There are various inverse problems — including reconstruction problems arising in medical imaging — where one is often aware of the forward operator that maps variables of interest to the observations. It is therefore natural to ask whether such knowledge of the forward operator can be exploited in deep learning approaches increasingly used to solve inverse problems.

In this paper, we provide one such way via an analysis of the generalisation error of deep learning methods applicable to inverse problems. In particular, by building on the algorithmic robustness framework, we offer a generalisation error bound that encapsulates key ingredients associated with the learning problem such as the complexity of the data space, the size of the training set, the Jacobian of the deep neural network and the Jacobian of the composition of the forward operator with the neural network. We then propose a `plug-and-play' regulariser that leverages the knowledge of the forward map to improve the generalization of the network. We likewise also propose a new method allowing us to tightly upper bound the Lipschitz constants of the relevant functions that is much more computational efficient than existing ones.  We demonstrate the efficacy of our model-aware regularised deep learning algorithms against other state-of-the-art approaches on inverse problems involving various sub-sampling operators such as those used in classical compressed sensing setup and accelerated Magnetic Resonance Imaging (MRI).
\end{abstract}
\section{Introduction}
\label{sec: Intro}
In various signal and image processing challenges arising in practice -- including medical imaging, remote sensing, and many more -- one often desires to recover a number of latent variables from physical measurements. This class of problems -- generally known as \emph{inverse problems} -- can often be modelled as follows:
\begin{equation}
\label{eq: Inverse_model}
    \mathbf{y}=\mathbf{A}\mathbf{x} + \mathbf{n}
\end{equation}
where $\mathbf{y} \in \mathcal{Y} \subset \mathbb{R}^q$ represents a $q$-dimensional vector containing the physical measurements, $\mathbf{x} \in \mathcal{X} \subset \mathbb{R}^p$ represents a $p$-dimensional vector containing the variables of interest, and $\mathbf{n}$ is a bounded perturbation modelling measurement noise (i.e. $\|\mathbf{n}\|\le \eta$). The forward operator modelling the relationship between physical measurements and variables of interests is in turn modelled (in the absence of noise) using a matrix $\mathbf{A} \in \mathbb{R}^{q \times p}$. This forward operator is also often assumed to satisfy certain regularity conditions such as $\Lambda_a$-Lipschitz continuity whereby 
\footnote{Note that such forward operators encountered in various applications of interest including Magnetic resonance Imaging (MRI), Computed Tomography (CT) etc obey some form of regularity constraint such as given in (\ref{eq: REC}).}
\begin{equation}
	\label{eq: REC}
	\|\mathbf{Ax}_1-\mathbf{Ax}_2\|_2\leq \Lambda_a\|\mathbf{x}_1-\mathbf{x}_2\|_2, \qquad \forall \mathbf{x}_1,\mathbf{x}_2 \in \mathcal{X}
	\end{equation}
Two broad classes of approaches have been adopted to solve inverse problems: (i) \textit{model-based} methods and (ii) \textit{data-driven} methods. Model-based methods exploit knowledge of the forward operator and/or the signal/noise model in order to recover the variables of interest from the measurements \cite{engl1996regularization}. For example, well-known inverse problem recovery algorithms often leverage knowledge of data priors capturing stochastic \cite{mueller2012linear} or geometric structure \cite{daubechies2004iterative}.
On the other hand, data-driven methods do not leverage explicitly the knowledge of the underlying physical and data models; instead, such methods rely on the availability of various data pairs $(\mathbf{x},\mathbf{y})$ in order to learn how to invert the forward operator associated with the inverse problem \cite{lucasusing}. 
The challenge relates to the fact that these approaches -- specially deep learning ones -- typically require the availability of various training examples that are not always available in a number of applications such as medical image analysis. This inevitably hinders the applicability of data-driven approaches to inverse problems arising in various scientific and engineering use-cases char.
Another emerging class of approaches tackles inverse problems by combining data-driven with model-based methods e.g. \cite{lunz2018adversarial,jin2017deep,bora2017compressed,wu2019deep,adler2018learned}.

In this paper, our overarching goal is to understand using first-principles how to use knowledge readily available in various inverse problems in order to improve the performance of deep learning based data-driven methods. We approach this challenge by offering new generalization guarantees that capture how the generalization ability is affected by various key quantities associated with the learning problem.
Such interplay then immediately leads to an entirely new model-aware regularization strategy acting as a proxy to import knowledge about the underlying physical model onto the deep learning process. 

Concretely, our contributions can be summarized as follows:
\begin{itemize}[leftmargin=.15in]
    \vspace{-0.25em}
    \item We present generalization error bounds for DNN based inverse problem solvers. Notably, such bounds depend on various quantities including the Jacobian matrix of the neural network along with the Jacobian matrix of the composition of the neural network with the inverse problem forward map.
    \vspace{-0.25em}
    \item We then propose new regularization strategies that are capable of using knowledge about the inverse problem model during the neural network learning process via the control of the spectral and Frobenius norms of such Jacobian matrices.
        \vspace{-0.25em}
    \item We also propose computationally efficient methods to estimate  the spectral and Frobenius norms of the aforementioned Jacobian matrices in order to accelerate the neural network learning process.%
        \vspace{-0.25em}
    \item Finally, we demonstrate the empirical performance of our algorithms on various inverse problems. These include the reconstruction of high-dimensional data from low-dimensional noisy measurements where the forward model is either a compressive random Gaussian matrix or a sub-sampling matrix usually employed in accelarated Magnetic Resonance Imaging applications.
\end{itemize}
The remainder of the paper is organized as follows: 
After presenting an overview of the related research in Section \ref{sec: related_work}, we introduce our system setup in Section \ref{sec: setup}. We then present generalization bounds applicable to neural network based inverse problem solvers in Section 4, leading up to model-aware regularizers in Section \ref{sec: discussion}.%
 Section \ref{sec: experiments} offers various experimental results showcasing our model-aware deep learning approach can lead to substantial gains in relation to model-agnostic ones. Finally, concluding remarks are drawn in Section \ref{sec: conclusion}. All the proofs, details of the experimental setup and additional results are relegated to the appendices.

{\textbf{Notation}}: We use lower case boldface characters to denote vectors, upper case boldface characters to represent matrices and sets are represented by calligraphic font. For example $\mathbf{x}$ is a vector, $\mathbf{X}$ is a matrix and $\mathcal{X}$ is a set. 
$\mathcal{N}_\mathcal{D}\left(\nicefrac{\psi}{2},\rho\right)$ represents the covering number of a metric space $(\mathcal{D},\rho)$ using balls of radius $\nicefrac{\psi}{2}$.
\section{Related Work}\label{sec: related_work}
Our work connects to various directions in the literature. 

\textbf{Data-driven techniques for inverse problems}: 
Deep learning techniques, inspired by their success in classification tasks, have been applied to a large number of inverse problems such as image denoising \cite{mao2016image, zhang2017beyond}, image super-resolution \cite{dong2016image}, MRI reconstruction \cite{hyun2018deep,lundervold2019overview}, CT reconstruction \cite{park2018computed}, and many more. However, these data-driven approaches typically require rich enough datasets -- which are not always available in various domains such as medical imaging -- in order to learn how to solve the inverse problem \cite{scruggs2015harnessing}.

\textbf{Model-aware data driven approaches}: In view of the fact that the underlying physical model is known in various scenarios, there are been an increased interest in model-aware data-driven approaches to inverse problems. Some approaches leverage knowledge of the forward model to provide a rough estimate of the inverse problem solution (e.g. using some form of pseudo-inverse of the forward operator) that is then further processed using a neural network \cite{jin2017deep, kang2017deep, han2018framing}
A recent (unsupervised) approach relies on the use of adversarially learnt data dependent regularizers \cite{lunz2018adversarial}: one can then formulate optimization problems containing a data fidelity term (where knowledge of the forward model is used) and the new regularization term (where the learnt data prior is exploited) in order to recover a solution of the underlying inverse problem.  

Another approach that is becoming increasingly popular relies on algorithm unfolding or unrolling \cite{gregor2010learning, monga2019algorithm}. By starting with a typical optimization based formulation to tackle the underlying inverse problem -- where knowledge of the physical model is explicitly used -- unfolding then maps iterative solvers onto a neural network architecture whose parameters can be further tuned in a data-driven manner.

Our work departs from these contributions in the sense that -- whereas we also use a deep network to solve an inverse problem -- we leverage knowledge of the underlying forward operator model via appropriate regularization strategies deriving from a principled generalization error analysis. 

\textbf{Other related work}: 
There is a considerable volume of literature offering analysis of the generalization ability of deep neural networks demonstrating that the generalization error of highly paramterized models can be bounded in terms of certain parameter norms \cite{neyshabur2018towards, bartlett2017spectrally}. 
Overall, the majority of these bounds are applicable to classification problems rather than regression based ones. Exceptions include \cite{amjad2018deep},
but their results suffered from an exponential dependence on network depth. Our current work addresses this issue, offering a study of the generalization ability of deep neural networks based inverse problems solvers, leading to entirely new gradient based regularization strategies allowing to incorporate knowledge into the learning process.%

Various works have already proposed approaches to efficiently introduce Lipschitz regularity in deep neural networks \cite{yoshida2017spectral, virmaux2018lipschitz, hoffman2019robust}. 
However, these techniques either do not take into account the non-linearities in the network \cite{yoshida2017spectral} or are only applicable to affine transformations \cite{virmaux2018lipschitz}.
We offer an algorithm to efficiently compute the spectral norm of the network Jacobian matrix. To the best of our knowledge, this is the tightest and most efficient manner to bound a deep neural network Lipschitz constant. \section{Setup}\label{sec: setup}
Our approach to solve the inverse problem in (\ref{eq: Inverse_model}) is based on the standard supervised learning paradigm. We assume access to a training set $\cS=\{\mathbf{s}_i=(\mathbf{x}_i,\mathbf{y}_i)\}_{i\leq m}$ consisting of $m$ data points $\mathbf{s}=(\mathbf{x},\mathbf{y})$ drawn independently and identically distributed (IID) from the sample space $\mathcal{D} = \mathcal{X} \times \cY$ according to the unknown data distribution $\mu$.  We also assume that $\mathcal{X}$ and $\mathcal{Y}$ are compact metric spaces with respect to $\ell_2$ metric and that the space $\mathcal{D} = \mathcal{X} \times \mathcal{Y}$ is compact with a product \footnote{Our analysis holds for most commonly used product metrics, such as sum, sup and $r$-norm product \cite{weaver1999lipschitz}. See Appendix \ref{sec: proofs}.} metric $\rho$. 

We use such a training set to learn a hypothesis $f_\mathcal{S}: \cY \rightarrow \mathcal{X}$ mapping the measurement variables to variables of interest. We then use such a hypothesis to map new measurement variables $\mathbf{y} \in \mathcal{Y}$ to the variables of interest $\mathbf{x} \in \mathcal{X}$ that were not necessarily originally present in the training set.

We restrict our attention to mappings based on feed-forward neural networks.
Such a feed forward neural network can be represented as a composition of $d$ layer-wise mappings delivering an estimate of the variable of interest given the measurement variable as follows:
\begin{IEEEeqnarray*}{rCl}
f_\mathcal{S}(\mathbf{y})=\left(f_{\theta_d}\circ\ldots f_{\theta_1}\right)(\mathbf{y};\Theta)
\end{IEEEeqnarray*}
where $f_{\mathcal{S}} (\cdot)$ represents the feed-forward neural neural network, $f_{\theta_i} (\cdot)$ represents the $i$-th layerwise mapping parameterized by $\theta_i$, and $\Theta=\{\theta_1,\ldots\theta_d\}$ is the set of tunable parameters in the neural network. 
The parameters of the feed-forward neural network are typically tuned based on the available training set using a learning algorithm such as stochastic gradient descent \cite{goodfellow2016deep}.

One is typically interested in the performance of the learnt neural network not only on the training data but also on (previously unseen) testing data. Therefore, it is useful to quantify the generalization error associated with the learn neural network given by:
\begin{equation}\label{GE}
	GE(f_\mathcal{S})=|l_{\text{exp}}(f_\mathcal{S})-l_{\text{emp}}(f_\mathcal{S})|
\end{equation}
where $l_{\text{exp}}(f_\mathcal{S})=\mathbb{E}_{\mathbf{s}\sim\mu}[l(f_\mathcal{S},\mathbf{s})]$ represents the expected error, $l_{\text{emp}}(f_\mathcal{S}) = \frac{1}{m}\sum_i l(f_\mathcal{S} ,\mathbf{s}_i)$ represents the empirical error, and the loss function $l: \mathbb{R}^p\times\mathbb{R}^{p}\rightarrow\mathbb{R}^+_0$ --- which measures the discrepancy between the neural network prediction and the ground truth --- is taken to be the $\ell_2$ distance given by:
 \begin{IEEEeqnarray}{lCr}
 l(f_\cS,\mathbf{s})=\|f_\cS(\mathbf{y})-\mathbf{x}\|_2
 \end{IEEEeqnarray}
Our ensuing analysis offers bounds to the generalization error in \eqref{GE} of deep feed-forward neural networks based inverse problems solvers as a function of a number of relevant quantities. These quantities include the covering number of the sample space $\mathcal{D}$, the size of the training set $\mathcal{S}$, and properties of the network encapsulated in its input-output Jacobian matrix
given by:
\begin{IEEEeqnarray*}{rCl}
\label{eq: J}
\mathbf{J}(\mathbf{y}) = \begin{bmatrix}
\frac{\partial f_\mathcal{S}(\mathbf{y})_1}{\partial \textbf{y}_{1}} & \cdots& \frac{\partial f_\mathcal{S}(\mathbf{y})_1}{\partial \textbf{y}_{q}}\\
\vdots & \ddots& \vdots\\
\frac{\partial f_\mathcal{S}(\mathbf{y})_{p}}{\partial \textbf{y}_{1}} & \cdots& \frac{\partial f_\mathcal{S}(\mathbf{y})_{p}}{\partial \textbf{y}_{q}}
\end{bmatrix}
\end{IEEEeqnarray*}
Our analysis will also inform how to import knowledge about the forward-operator associated with the inverse problem onto the learning procedure.
\section{Analysis: Generalization Error Bounds}\label{sec: theory}
Our analysis builds upon the \emph{algorithmic robustness} framework in \cite{xu2012robustness}.
	\begin{defn}\label{def:robust}
		Let $\cS$ and $\mathcal{D}$ denote the training set and sample space. A learning algorithm is said to be $(K, \epsilon(\cS))$-robust if the sample space $\mathcal{D}$ can be partitioned into $K$ disjoint sets $\mathcal{K}_k$, $k = 1,\ldots ,K$,	such that for all $(\mathbf{x}_i,\mathbf{y}_i)\in \cS$ and all $(\mathbf{x},\mathbf{y})\in\mathcal{D}$
		\begin{align}
		\mathbf{s}_i, \mathbf{s}\in\mathcal{K}_k\implies \left|l(f_\mathcal{S},\mathbf{s}_i)-l(f_\mathcal{S},\mathbf{s})\right|\leq \epsilon(\cS) \label{eq:epsilon}
		\end{align}
	\end{defn}
This notion has already been used to analyse the performance of deep neural networks in \cite{sokolic2017robust,cisse2017parseval, jia2019orthogonal}. However, such analyses applicable to classification tasks do not carry over immediately to inverse problems based tasks where -- in addition to using knowledge about the forward model associated with the inverse problem -- there are some technical complications deriving from the fact that the loss functions are typically unbounded. \footnote{Existing work applies to uniformly bounded loss function (e.g. \cite{xu2012robustness,sokolic2017robust}).} 

We begin addressing these challenges by offering a simple result that showcases how the distance between the neural network estimates of the variables of interest depends on the distance between the variables of interest themselves and, importantly, the Jacobian of the network, the Jacobian of the composition of the network with the forward model associated with the inverse problem, and the noise power associated with the inverse problem.
		\begin{thm}
	 \label{thm: lipschitz_cont_end2end}
	 Consider a neural network $f_{\cS} (\cdot):\cY\rightarrow\mathcal{X}$ based solver of the inverse problem in (\ref{eq: Inverse_model}), learnt using a training set $\mathcal{S}$. 
	 Then, for any $\mathbf{s}' = (\mathbf{x}',\mathbf{y}'),~\mathbf{s}'' = (\mathbf{x}'',\mathbf{y}'')\in\mathcal{D} = \mathcal{X} \times \mathcal{Y}$, it follows that 
		\begin{IEEEeqnarray*}{rCl}
		\|{f_\mathcal{S}(\mathbf{y}'')-f_\mathcal{S}(\mathbf{y}')}\|_2\le\Lambda_{f\circ a}\|{\mathbf{x}''-\mathbf{x}'}\|_2+2\eta\Lambda_{f}
		\end{IEEEeqnarray*}
		where $\Lambda_{f\circ a}$ and $\Lambda_f$ are upper bounds to the Lipschitz constants of the neural network and the composition of the neural network and the forward operator respectively.
		\begin{IEEEeqnarray}{rCl}
		\label{eq: lipschitz_constants}
		\Lambda_{f\circ a}&=&\sup_{\mathbf{y}\in conv(\mathcal{Y})}\|{\mathbf{J}\rbr{\mathbf{y}}\mathbf{A}}\|_2\quad \Lambda_f=\sup_{\mathbf{y}\in conv(\mathcal{Y})}\|\mathbf{J}\rbr{\mathbf{y}}\|_2
		\end{IEEEeqnarray} 
	\end{thm}
We now state another theorem -- building upon Theorem 1 -- articulating about the robustness of a deep neural network based solver of an inverse problem.
	\begin{thm}\label{thm: robustness}
		Consider that $\mathcal{X}$ and $\cY$ are compact spaces with respect to the $\ell_2$-metric. Consider also the sample space $\mathcal{D}=\mathcal{X}\times\cY$ equipped with a product metric $\rho$. 
It follows that a neural network trained to solve an inverse problem in (\ref{eq: Inverse_model}) based on a training set $\cS$ is
		\begin{equation*}
		\left(\cN_\cX\left(\nicefrac{\delta}{2},\ell_2\right),\left(1+\Lambda_{f\circ a}\right)\delta+2\eta\Lambda_f\right)-\text{robust}
		\end{equation*}
		for any $\delta>0$ and $\cN_\cX\left(\nicefrac{\delta}{2},\rho\right) < \infty$.
	\end{thm}
We now state our main theorem relating to the generalization error of a deep neural network trained to solve an inverse problem.
\begin{thm}(GE Bound)\label{thm: ge}
	Consider again that $\mathcal{X}$ and $\cY$ are compact spaces with respect to the $\ell_2$ metric. Consider also the sample space $\mathcal{D}=\mathcal{X}\times\cY$ equipped with a product metric $\rho$. It follows that a neural network trained to solve the inverse problem in (\ref{eq: Inverse_model}) based on a training set $\cS$ consisting of $m$ i.i.d. training samples obeys with probability $1-\zeta$, for any $\zeta>0$, the $GE$ bound given by:
	\begin{IEEEeqnarray*}{rCl}
		\nonumber GE(f_\mathcal{S})
		 \le \left(1+\Lambda_{f\circ a}\right)\delta+2\eta\Lambda_{f}+M\sqrt{\frac{2\cN_\cX\left(\nicefrac{\delta}{2},\ell_2\right)\log(2)+2\log\left(\nicefrac{1}{\zeta}\right)}{m}}
	\end{IEEEeqnarray*}
	for any $\delta>0$ and $M<\infty$. 
	\end{thm}
One can derive various insights from this theorem that is applicable to any differentiable feed forward neural network along with any Lipschitz continuous forward map %
: (1) first, in line with traditional bounds \cite{neyshabur2018towards,wei2019data}, the generalization error depends on the size of training set $\mathcal{S}$; (2) second, in line with more recent bounds \cite{sokolic2017robust,cisse2017parseval, amjad2018deep}, the generalization error also depends on the complexity of the data space $\mathcal{D}$;
(3) Finally, Theorem \ref{thm: ge} also reveals that the operator norm of the Jacobian of the network and the composite map also play a critical role: the lower the value of these norms, the lower the generalization error. More importantly, the proposed generalization bound is also non-vacuous in the network parameters because the network Jacobian matrix does not directly depend on the network depth. This is in sharp contrast with existing generalization bounds that typically depends on the network depth \cite{neyshabur2018towards,cisse2017parseval}
\section{ Model-Aware Jacobian Regularization}\label{sec: discussion}
Our approach to leverage knowledge about the inverse problem model onto the learning process involves regularization.
In particular, Theorem \ref{thm: ge} suggests that penalizing the spectral norm of the Jacobian of the neural network and the spectral norm of the Jacobian of the composition of the neural network with the inverse problem forward operator,
which incidentally also serve as an upper bound to the Lipschitz constants of these mappings, should improve the generalization ability of a neural network based inverse problem solver.

The use of Lipschitz regularization to improve the generalization ability of deep neural networks has already been recognized by various works \cite{novak2018sensitivity, cisse2017parseval, sokolic2017robust, hoffman2019robust}. However, the fact that introducing Lipschitz regularity in the-end to-end mapping composed of the neural network and the inverse problem forward map may also control generalization does not appear to have been acknowledged in previous works pertaining to deep learning approaches to inverse problems. 

{\bf Model-Aware Spectral Norm Based Regularization}: 
Our first regularization strategy directly penalizes the operator norm of the Jacobians for the neural network and of the composition of the neural network and the forward map.
 Training in a minibatch stochastic gradient setup, where the optimization is carried out over minibatches $\mathcal{B} = \{\textbf{s}_1, \textbf{s}_2,\ldots, \textbf{s}_{|\mathcal{B}|}\}$, leads to the following objective 
\begin{equation}
\label{eq: loss_spectral}
   \frac{1}{|\mathcal{B}|}\sum_{i=1}^{|\mathcal{B}|}l(f_\mathcal{S},\mathbf{s}_i)+\lambda_1\max_{\mathbf{s}\in\mathcal{B}}\|\mathbf{J}(\mathbf{y})\mathbf{A}\|_2+\lambda_2\max_{\mathbf{s}\in\mathcal{B}}\|\mathbf{J}(\mathbf{y})\|_2
\end{equation}
where $\lambda_1,\lambda_2$ are hyper-parameters.

{\bf Model-Aware Frobenius Norm Based Regularization:}
Our second regularization strategy stems from the fact that the Frobenius norm upper bounds the Spectral norm. Regularisation strategies that punish the Frobenius norm of the network Jacobian have been associated with significant improvement in robustness of DNN classifiers \cite{novak2018sensitivity, sokolic2017robust, hoffman2019robust}. Therefore, we also propose the following cost function
\begin{equation}
\label{eq: loss_frob}
  \frac{1}{|\mathcal{B}|}\sum_{i=1}^{|\mathcal{B}|}l(f_\mathcal{S},\mathbf{s}_i)+\lambda_1\sum_{i=1}^{|\mathcal{B}|}\|\mathbf{J}(\mathbf{y}_i)\mathbf{A}\|_F^2+\lambda_2\sum_{i=1}^{|\mathcal{B}|}\|\mathbf{J}(\mathbf{y}_i)\|_F^2
\end{equation}
Note that $\lambda_2=0$ in a noise free setting.

{\bf Efficient Computation of the Norms of the Jacobian Based Regularizers: }\begin{algorithm2e}[t]
\small
\KwIn{ Mini-batch $\cB$,number of power iterations $n$.}
\KwOut{Maximum singular value, $\sigma$, of $\mathbf{J}$.}
\For {$(\mathbf{y},\mathbf{x})\in\cB$}
{
    Initialize $\left\{\mathbf{v}^{(0)}\right\} \sim \mathcal{N}(0, \mathbf{I})$\\
    $i \leftarrow 1$\\
    \While{ $i\le n$ }
    {
        $\mathbf{u}^{(i)} \leftarrow jvp(f(\mathbf{y}),\mathbf{y},\mathbf{v}^{(i-1)})$\\
        $\mathbf{v}^{(i)} \leftarrow vjp(f(\mathbf{y}), \mathbf{y}, \mathbf{u}^{(i)})$ \\
         
        $i \leftarrow i+1$.
    }$\sigma\leftarrow\nicefrac{\|\mathbf{u}\|_2}{\|\mathbf{v}\|_2}$} 
    \caption{
        Estimation of the Spectral norm of the Jacobian matrix.
    }
    \label{alg:  spectral_j}
\end{algorithm2e} %
The challenge associated with the use of the training objectives in (\ref{eq: loss_spectral}) and (\ref{eq: loss_frob}) relates to the computation of the Spectral norm and Frobenious norm of both $\mathbf{J}$ and $\mathbf{JA}$ because computing and storing the Jacobian matrix of deep neural networks incurs huge cost. The random projection based method proposed in \cite{hoffman2019robust} -- which can be used to approximate the square of the Frobenius norm of a matrix-- can be immediately extended to approximate the regularization terms in (\ref{eq: loss_frob}), but the technique cannot be used to estimate the regularization terms in (\ref{eq: loss_spectral})

We therefore also offer a new computationally efficient method – which can be easily integrate dwith modern deep learning libraries such as Tensorflow \cite{abadi2016tensorflow} – based on the power method \cite{mises1929praktische}.
In particular let the maximum singular value and the corresponding left and right singular vectors of the matrix be denoted by $\sigma$, $\mathbf{u}$ and $\mathbf{v}$ respectively. Then  starting from a randomly initialized $\mathbf{v}^{(0)}_1\in\mathbb{R}^{q}$, the power method performs a simple recursive routine applied on a matrix $\mathbf{J}$ using a finite number of iterations $n$ ($n=3$ in our setup) as follows:
\begin{IEEEeqnarray}{rCl}
\label{eq: power_alg}
\mathbf{u}^{(i)}\leftarrow\mathbf{J}\mathbf{v}^{(i-1)},\qquad\mathbf{v}^{(i)}\leftarrow\mathbf{J}^T\mathbf{u}^{(i)}
\end{IEEEeqnarray}
This leads to an approximation of the spectral norm of the matrix as follows $\sigma=\nicefrac{\|\mathbf{u}^{(n)}\|_2}{\|\mathbf{v}^{(n)}\|_2}$. The fact that this algorithm is simple and accurate has led to its usage in regularization approaches involving the spectral norm of the weight matrices \cite{bansal2018can, yoshida2017spectral, virmaux2018lipschitz}. 

However, our regularization approach involves the relevant Jacobian matrices whose determination is both computationally and memory intense even in a low-dimensional setting. As most optimization algorithms rely on gradient based updates, modern deep learning libraries borrow tools from the field of automatic differentiation to efficiently compute the reverse mode vector Jacobian product i.e, a Jacobian matrix left multiplied by a vector -- abbreviated to ${vjp}$ -- as follows:
\begin{IEEEeqnarray}{rCl}
\label{eq: vjp}
vjp(f_\mathcal{S}(\mathbf{y}),\mathbf{y},\mathbf{d})=\mathbf{d}^T\mathbf{J}
\end{IEEEeqnarray}
where $\mathbf{d}$ is a user specified weighting vector that is set to all ones by default. We note that to be able to compute spectral norm of the Jacobian matrix, we need to iteratively compute ${vjp}$ and ${jvp}$ -- the forward mode Jacobian vector product. Although most popular ML libraries do not provide the support to calculate the $jvp$ directly, it can in fact be efficiently computed using the existing functionality \cite{townsend2017}. We leverage this technique to propose an automatic differentiation compliant algorithm which uses the power method to determine the spectral norm of the Jacobian of a differentiable function without the need to compute the actual Jacobian matrix itself. This is summarized in Algorithm \ref{alg:  spectral_j}. We provide a complexity analysis based on compute time and memory in the Appendix \ref{sec: acc_jvp}.
\section{Experiments}\label{sec: experiments}
\begin{figure}[t]
\begin{subfigure}{0.60\textwidth}
    \centering
    \includegraphics[width=0.49\linewidth]{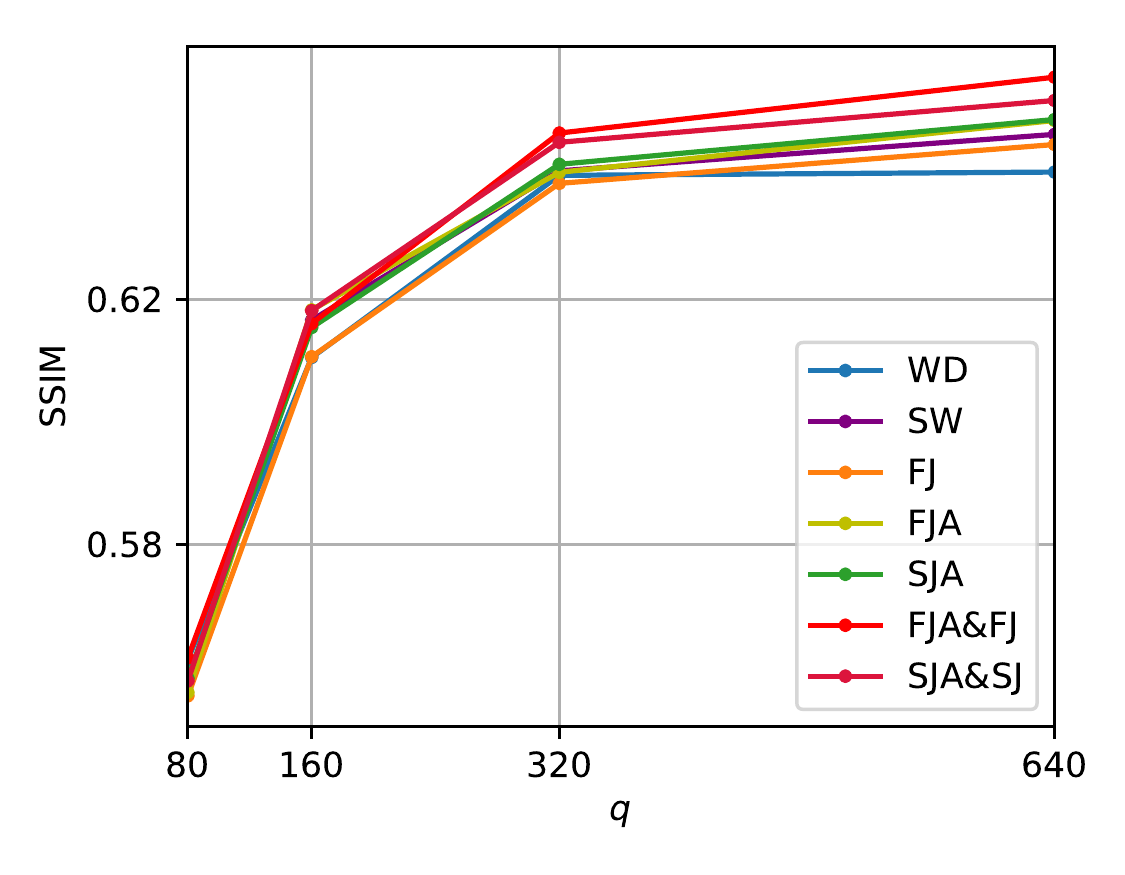}
    \includegraphics[width=0.49\linewidth]{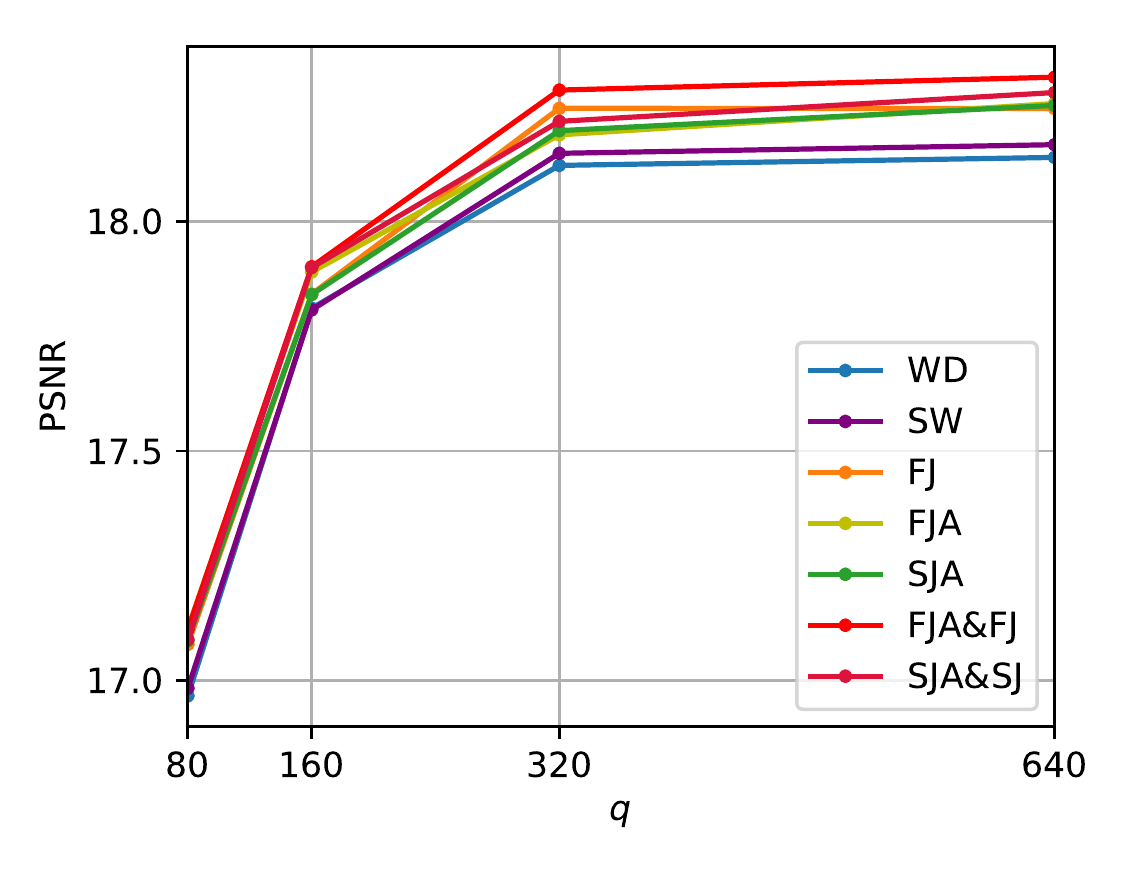}
    \vspace{-0.25em}
    \caption{
    }
\label{fig: SSIM_PSNR}
\end{subfigure}
\begin{subfigure}{0.35\textwidth}
  \centering
      \includegraphics[width=0.15\textwidth]{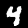}
    \includegraphics[width=0.15\textwidth]{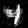}
    \includegraphics[width=0.15\textwidth]{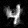}
    \includegraphics[width=0.15\textwidth]{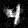}
    \includegraphics[width=0.15\textwidth]{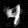}
    \includegraphics[width=0.15\textwidth]{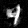} 
    \\
    \vspace{0.1em}
    \includegraphics[width=0.15\textwidth]{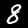}
    \includegraphics[width=0.15\textwidth]{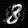}
    \includegraphics[width=0.15\textwidth]{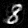}
    \includegraphics[width=0.15\textwidth]{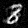}
    \includegraphics[width=0.15\textwidth]{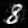}
    \includegraphics[width=0.15\textwidth]{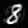}  \\
    \vspace{0.1em}
    \includegraphics[width=0.15\textwidth]{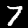}
    \includegraphics[width=0.15\textwidth]{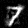}
    \includegraphics[width=0.15\textwidth]{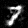}
    \includegraphics[width=0.15\textwidth]{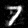}
    \includegraphics[width=0.15\textwidth]{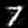}
    \includegraphics[width=0.15\textwidth]{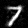}  \\
    \caption{ }
\label{fig: Vis}
\end{subfigure}
\caption{(a): Reconstruction of MNIST images given Gaussian compressive measurements using a fully connected neural network. The LHS reports SSIM versus number of Gaussian measurements and the RHS reports PSNR versus number of Gaussian measurements.(b) Sample results from the reconstruction from Gaussian measurements in a noise free setting. (Top to bottom) $q=40,80,320$. (Left to Right) ground truth, reconstruction using Adam optimizer with WD, WS, FJ, FJA, SJA regularization.
}
\label{fig:image2}
\end{figure}
In this section, we evaluate the effectiveness of the proposed Jacobian regularization on different inverse problems. 

In all the figures and tables below, we refer to regularization strategies appearing in (\ref{eq: loss_spectral}) and (\ref{eq: loss_frob}) as SJA\&SJ, FJA\&FA ($\lambda_1,\lambda_2>0$) and SJA, FJA ($\lambda_1>0, \lambda_2=0$). We provide comparisons with benchmark schemes in terms of visualizations and quality metrics such as Structural Similarity Index (SSIM) and Peak Signal to Noise Ratio (PSNR). The details of the experimental setup and results is summarized in the sequel.
\subsection{Gaussian Measurements}
\label{sec: gaussian_experiments}
{\bf Experimental Setup:} In order to show that including the knowledge of the forward map results in performance gains over model agnostic data driven methods, we consider a simple compressive sensing setting where $\mathbf{A}$ is a (wide) random Gaussian matrix. Each entry in this matrix is sampled IID from a zero mean Gaussian distribution of variance $\nicefrac{1}{q}$. The noise level $\eta$ takes values in the set $\{0,0.3\}$. To observe how regularizers behave in presence of small number of measurements, we construct multiple operators $\mathbf{A}$ with their number of rows, $q$ equal to $40$, $80$, $160$, $320$ and $640$.
The ground truth in this setting is sampled from the MNIST dataset \cite{lecun-mnisthandwrittendigit-2010}. 
We fix the train size to $500$ and normalize the training labels before applying the linear transform $\mathbf{A}$ and random noise.

For reconstruction, we use fully connected networks which consist of an input layer of size $q$ neurons, followed by three layers of width $p$. All the layers except the last one have an associated ReLU activation function. 

 {\bf Results:} 
Fig \ref{fig: SSIM_PSNR}, depicts the comparative performance of networks regularized with our Jacobian regularizers and the baseline techniques weight decay (WD), spectral norm regularization of weights (WS) \cite{miyato2018spectral}, frobenius norm regularization of Jacobian (FJ) \cite{hoffman2019robust} on various measurement lengths with the performance gains being more pronounced as the measurement size is increased.
A visual comparison of the quality of the reconstructed images in a noise free setting is presented in Fig \ref{fig: Vis}. It can be seen that the quality of images, recovered with our proposed SJA and FJA regularization is perceptually more refined for different number of measurements. 

These results support our analysis that model induced regularizers improve the performance of the DNN over model agnostic regularization translating into reconstructions with better SSIM and PSNR.
\subsection{k-space subsampled measurements}
\begin{figure}[t]
    \centering
    \vspace{0.1em}    \includegraphics[width=0.13\textwidth]{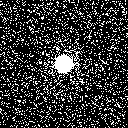} 
    \includegraphics[width=0.13\textwidth]{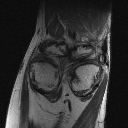} 
    \includegraphics[width=0.13\textwidth]{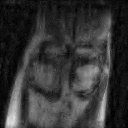}
     \includegraphics[width=0.13\textwidth]{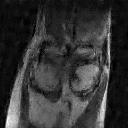} 
    \includegraphics[width=0.13\textwidth]{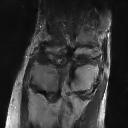} 
    \includegraphics[width=0.13\textwidth]{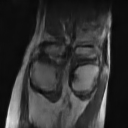}
      \includegraphics[width=0.13\textwidth]{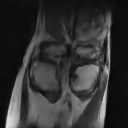}
    \\
    \vspace{0.1em}
    \includegraphics[width=0.13\textwidth]{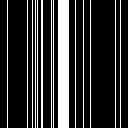}
   \includegraphics[width=0.13\textwidth]{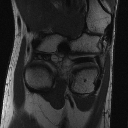}
   \includegraphics[width=0.13\textwidth]{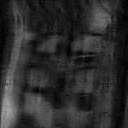}
    \includegraphics[width=0.13\textwidth]{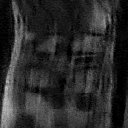} 
   \includegraphics[width=0.13\textwidth]{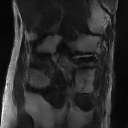}
   \includegraphics[width=0.13\textwidth]{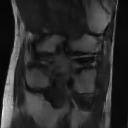}
   \includegraphics[width=0.13\textwidth]{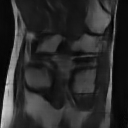}
    \\
    \vspace{0.2em}
    \caption{
        Sample results from the reconstruction from k-space subsampled measurements for different types of acquisition masks.(Left to Right) sampling mask, ground truth, reconstruction using $\ell_1$ Wavelet regularization, Adversarial Regularizer \cite{lunz2018adversarial}, postprocessing using UNet \cite{jin2017deep}, UNet with FJA\&FJ and UNet with SJA\&SJ. 
    }
    \label{fig: mri_recon}
    \vspace{-1em}
\end{figure}
\label{sec: subsampled_experiments}
{\bf{Experimental Setup:}} For our next set of experiments, we consider 
the frequency domain sub-sampling operator appearing in MR imaging. It can be mathematically represented as
$\mathbf{A}=\mathbf{F}^{-1}\mathbf{M}\mathbf{F}$. Here $\mathbf{F}$ and $\mathbf{F}^{-1}$ are the 2D Fourier and inverse Fourier transform matrices. The mask $\mathbf{M}$ is diagonal matrix containing binary entries on its diagonal where the fraction of non-zero entries signify the subsampling ratio $s$. 

We generate the training and validation set by retrospectively under-sampling the Fourier transform of the ground truth images, obtained from the NYU fastMRI's knee database \cite{zbontar2018fastmri}. The subsampling is achieved by the Cartesian 1D and 2D random sampling masks in k-sapce, retaining only $25\%$ and $20\%$ of the total Fourier samples respectively. We normalize the images to the range $[0,1]$ before applying the forward transform\footnote{Additional details of the preprocessing, network architecture and the training routine will be included in the Appendix \ref{sec: experiments_apndx}.} and fix the noise level to $\eta=5$.

These noisy subsampled measurements are fed to UNet architecture \cite{ronneberger2015u} for reconstruction which we train using Adam optimizer for $300$ epochs. We use a training set size of $500$ and a minibatch of size $5$. All our results are achieved by applying the regularization at only a small fraction ($10\%$) of steps per epoch. The regularization coefficients $\lambda_1$ and $\lambda_2$ are tuned in an adaptive manner during training for which the details are provided in the Appendix \ref{sec: experiments_apndx}.

{\bf Results:} 
\begin{table}[t]
\scriptsize
    \centering
    \caption{\small
        Summary of MRI image reconstruction using UNet.
    }
    \vspace{0.5em}
    \begin{tabular}{r c rr c rr c rr}
        \toprule
                && \multicolumn{2}{c}{2D mask ($s=0.2$)} && \multicolumn{2}{c}{1D mask ($s=0.25$)}\\
                \cmidrule{3-4} \cmidrule{6-7}
                && PSNR & SSIM && PSNR & SSIM\\
        \midrule
        Adversarial Regularizer &&  29.89 &   0.77 &&    25.44 &   0.54\\
        UNet &&  29.10 &   0.76 &&    28.81 &   0.74\\
        UNet as postprocessing &&  30.01 &    0.79 && 28.36    &    0.74 \\
        UNet w FJA\&FJ &&   30.80 &    0.80 &&    28.96 &    0.75 \\
        UNet w SJA\&SJ &&   30.89 &    0.81 &&      29.30 &    0.78\\
        \bottomrule
    \end{tabular}
    \label{tab: mri}
\end{table}
In this setting, we compare the reconstruction results achhieved through the proposed FJA\&FJ and SJA\&SJ regularization with the $\ell_1$-regularized wavelet reconstruction, Adversarial Regualarization method introduced in \cite{lunz2018adversarial} and the UNet postprocessing method \cite{jin2017deep}. For a fair comparison, the UNet architecture and training routines are kept the same for our work and the postprocessing method. 
For the Adversarial regularization method, we used the implementation provided by the authors on Github and kept the batch size and training routine unchanged. For both, Adversarial Regularization and postprocessing method, we use the reconstruction of $\ell_1$ regularized method as input.

A qualitative comparison of the performance for various experimental settings is provided in Table \ref{tab: mri}. The proposed regularizers consistently outbeat all the other methods in terms of PSNR and SSIM. The performance gains are more pronounced for 1D sampling mask which introduces aliasing artifacts in the measurement. 
A visual inspection of the achieved reconstructions is provided in Figure \ref{fig: mri_recon}. A close inspection of the reconstructed images reveals that the proposed method introduces less artifacts than the other reconstructions.
\section{Conclusion}\label{sec: conclusion}
This paper -- leveraging knowledge of underlying physical models -- proposes a new deep learning approach to solve inverse problems. The crux of the approach -- stemming directly from a rigorous generalization error analysis -- is a new neural network learning procedure involving the use of cost functions in capturing knowledge of the underlying inverse problem model via appropriate regularization. This regularizer, owing to its plug-and-play nature can be integrated into any deep learning based solver of inverse problems without extra hassle. Empirical results on a variety of problems have shown that our proposed regularization approach can outperform considerably standard model agnostic regularizers and reconstruction schemes specialized for inverse problems. This work adds to recent ones by showing there is much value incorporating model knowledge onto data-driven approaches. 
\section{Broader Impact}
The outstanding performance offered by deep neural networks to long-standing problems has encouraged its use in a myriad of applications. However, deep neural networks -- often called \emph{black boxes} -- are poorly understood, leading to predictions that are often not interpretable or explainable.
While in certain application fields this issue may play a secondary role, in high-risk domains, e.g., healthcare, it is crucial to use machine learning models that are trustworthy.

In this work, we take a step in this direction by providing a framework for explaining the various factors that affect the performance of a deep neural network on a well-known class of problems: inverse problems. This is an important class of problems that arises in various scientific and engineering applications including imaging techniques widely used in healthcare. We offer a principled methodology offering the means to train more robust deep neural network models.

However, our methodology also has a cost. It relies on regularization techniques resulting in additional training time hence an increased carbon footprint. This can nonetheless be partially mitigated by using the proposed strategy on a limited number of training steps as apposed to each training step. 
\bibliographystyle{unsrt}
\bibliography{main}
\setcounter{equation}{10}
\setcounter{figure}{2}
\setcounter{algocf}{1}

\newpage
\appendix
\appendixpage
\section{Proofs}
\label{sec: proofs}
\begin{proof}[Proof of Theorem \ref{thm: lipschitz_cont_end2end}]
We first note that the line between $\mathbf{y}_1 = \mathbf{A} \mathbf{x}_1 + \mathbf{n}_1$ and $\mathbf{y}_2 = \mathbf{A} \mathbf{x}_2 + \mathbf{n}_2$ is given by $
\bar{\theta}\mathbf{y_1}+\theta\mathbf{y}_2$ where $\theta\in(0,1)$ and $\bar{\theta}=1-\theta$. Let us now define a function $h(\theta)$ as follows:
\begin{IEEEeqnarray*}{rCl}
h(\theta)&=&f_\cS(\bar{\theta}\mathbf{y_1}+\theta\mathbf{y}_2) = f_\mathcal{S}\rbr{\mathbf{A}\rbr{\bar{\theta}\mathbf{x}_1+\theta\mathbf{x}_2}+\bar{\theta}\mathbf{n}_1+\theta\mathbf{n}_2}    
\end{IEEEeqnarray*}
By the generalized fundamental theorem of calculus, it can be shown that: 
\begin{IEEEeqnarray*}{rCl}
f_\mathcal{S}(\mathbf{y}_2)-f_\mathcal{S}(\mathbf{y}_1)=\int_{0}^{1}\frac{d h (\theta)}{d\theta}d\theta
\end{IEEEeqnarray*}
where
	\begin{IEEEeqnarray*}{rCl}
	\frac{d}{d\theta}(h(\theta))=\mathbf{J}\rbr{\bar{\theta}\mathbf{y_1}+\theta\mathbf{y}_2}\sbr{\mathbf{A}\rbr{\mathbf{x}_2-\mathbf{x}_1}+\rbr{\mathbf{n}_2-\mathbf{n}_1}}
	\end{IEEEeqnarray*}
Then, from the sub-multiplicative property of matrix norms, it is immediate to show that:
	\begin{IEEEeqnarray*}{rCl}
	&&|{f_\mathcal{S}\rbr{\mathbf{y}_2}-f_\mathcal{S}\rbr{\mathbf{y}_1}}\|_2\\
	&=&\left\|\int_{0}^{1}\mathbf{J}\rbr{\bar{\theta}\mathbf{y_1}+\theta\mathbf{y}_2}\sbr{\mathbf{A}\rbr{\mathbf{x}_2-\mathbf{x}_1}+\rbr{\mathbf{n}_2-\mathbf{n}_1}}d\theta\right\|_2\\
	&\le&\left\|\int_{0}^{1}\mathbf{J}\rbr{\bar{\theta}\mathbf{y_1}+\theta\mathbf{y}_2}\mathbf{A}({\mathbf{x}_2-\mathbf{x}_1})d\theta\right\|_2+\left\|{\int_{0}^{1}\mathbf{J}\rbr{\bar{\theta}\mathbf{y_1}+\theta\mathbf{y}_2}({\mathbf{n}_2-\mathbf{n}_1})d\theta}\right\|_2\\
	&\le&\left\|\int_{0}^{1}\mathbf{J}\rbr{\bar{\theta}\mathbf{y_1}+\theta\mathbf{y}_2}\mathbf{A}d\theta\right\|_2\|{\mathbf{x}_2-\mathbf{x}_1}\|_2+\left\|{\int_{0}^{1}\mathbf{J}\rbr{\bar{\theta}\mathbf{y_1}+\theta\mathbf{y}_2}d\theta}\right\|_2\|{\mathbf{n}_2-\mathbf{n}_1}\|_2
	\end{IEEEeqnarray*}  

It is also possible to show that:
	\begin{IEEEeqnarray*}{rCl}
\left\|\int_{0}^{1}\mathbf{J}\rbr{\bar{\theta}\mathbf{y_1}+\theta\mathbf{y}_2}\mathbf{A}d\theta\right\|_2&\stackrel{(a)}{\le}& \int_{0}^{1}\|\mathbf{J}\rbr{\bar{\theta}\mathbf{y_1}+\theta\mathbf{y}_2}\mathbf{A}\|_2d\theta\\
	&{\le}& \sup_{\stackrel{\mathbf{y}_1,\mathbf{y}_2\in\mathcal{Y}} {\theta\in[0,1]}}\|\mathbf{J}\rbr{\bar{\theta}\mathbf{y_1}+\theta\mathbf{y}_2}\mathbf{A}\|_2
	\end{IEEEeqnarray*}

Therefore, given that $\bar{\theta}\mathbf{y}_1+\theta\mathbf{y}_2$ is in convex-hull of $\mathcal{Y}$ for $\theta\in[0,1]$, it follows immediately that:
\begin{IEEEeqnarray}{rCl}
\label{eq: Jacobian_conv_}
&&\nonumber \|{f_\mathcal{S}\rbr{\mathbf{y}_2}-f_\mathcal{S}\rbr{\mathbf{y}_1}}\|_2\\
&\le&\sup_{\mathbf{y}\in conv(\mathcal{Y})}\|{\mathbf{J}\rbr{\mathbf{y}}\mathbf{A}}\|_2\|{\mathbf{x}_2-\mathbf{x}_1}\|_2+\sup_{\mathbf{y}\in conv(\mathcal{Y}}\|{\mathbf{J}\rbr{\mathbf{y}}}\|_2\|{\mathbf{n}_2-\mathbf{n}_1}\|_2\nonumber\\
&\le&\sup_{\mathbf{y}\in conv(\mathcal{Y})}\|{\mathbf{J}\rbr{\mathbf{y}}\mathbf{A}}\|_2\|{\mathbf{x}_2-\mathbf{x}_1}\|_2+2\eta\sup_{\mathbf{y}\in conv(\mathcal{Y})}\|{\mathbf{J}\rbr{\mathbf{y}}}\|_2
\end{IEEEeqnarray}
where $conv(\mathcal{Y})$ represents the convex hull of $\mathcal{Y}$.
\end{proof}
	\begin{proof}[Proof of Theorem \ref{thm: robustness}]
We first prove the following Lemma showing that for a Lipschitz continuous linear mapping $\mathbf{A}$, it is possible to upper bound the covering number of $\mathcal{D}$ via the covering number of $\mathcal{X}$. 
\begin{lem}
    Let $\mathbf{A}\in\mathbb{R}^{p \times q}$ be a $\Lambda_a$-Lipschitz continuous linear map from the compact metric space $(\mathcal{X},\ell_2)$ to $(\mathcal{Y},\ell_2)$ defined in (\ref{eq: Inverse_model}). Then, assuming that the product space $\mathcal{D}=\mathcal{X}\times\mathcal{Y}$ is compact with respect to
\begin{enumerate}
     \item[i - ]  $r$-norm product metric $\rho=\left(\|\mathbf{y}-\mathbf{y}'\|^r_2+\|\mathbf{x}-\mathbf{x}'\|_2^r\right)^{\nicefrac{1}{r}}$
    \begin{IEEEeqnarray*}{rCl}
    \mathcal{N}_\mathcal{D}(\left(\delta^r+(\Lambda_{a}\delta+\eta)^r\right)^{\nicefrac{1}{r}},\rho)\leq\mathcal{N}_\mathcal{X}(\nicefrac{\delta}{2},\ell_2)
    \end{IEEEeqnarray*}
    \item[ii - ]  sup product metric $\rho=\max(\|\mathbf{y}-\mathbf{y}'\|_2,\|\mathbf{x}-\mathbf{x}'\|_2)$
    \begin{IEEEeqnarray*}{rCl}
    \mathcal{N}_\mathcal{D}(\max(\delta,\Lambda_{a}\delta+\eta),\rho)\leq\mathcal{N}_\mathcal{X}(\nicefrac{\delta}{2},\ell_2)
    \end{IEEEeqnarray*}
     \item[iii - ] sum product $\rho=\|\mathbf{y}-\mathbf{y}'\|_2+\|\mathbf{x}-\mathbf{x}'\|_2$
     \begin{IEEEeqnarray*}{rCl}
    \mathcal{N}_\mathcal{D}((1+\Lambda_a)\delta+\eta,\rho)\leq\mathcal{N}_\mathcal{X}(\nicefrac{\delta}{2},\ell_2)
    \end{IEEEeqnarray*}
\end{enumerate}
\label{lem: product_sp_covnumber}
\end{lem}
\begin{proof}	
We will first show that if the set $\mathcal{X}'$  is a $\delta$-cover of the set $\mathcal{X}$ with respect to the $\ell_2$ metric then the set $\mathcal{Y}' = \{\mathbf{y}' = \mathbf{A}\mathbf{x}', \mathbf{x}'\in\mathcal{X}'\}$ is a $(\Lambda_a\delta+\eta)$-cover of $\mathcal{Y}$ with respect to the $\ell_2$ metric. In particular, in view of the fact that since $\forall \mathbf{x}\in\mathcal{X}$, $\exists\mathbf{x}'\in\mathcal{X}'$ such that $\|\mathbf{x}-\mathbf{x}'\|_2\leq{\delta}$ then $\forall \mathbf{y}\in\mathcal{Y}$, $\exists\mathbf{y}'\in\mathcal{Y}'$ such that:
	\begin{IEEEeqnarray*}{rCl}
		\|\mathbf{y}-\mathbf{y}'\|_2&\le&\|\mathbf{Ax}+\mathbf{n}-\mathbf{A}\mathbf{x}'\|_2\\
		&\leq&\|\mathbf{Ax}-\mathbf{A}\mathbf{x}'\|_2+\|\mathbf{n}\|_2\\
		&\leq& \Lambda_a\|\mathbf{x}-\mathbf{x}'\|_2+\|\mathbf{n}\|_2\\
		&\leq& \Lambda_a\delta+\eta
	\end{IEEEeqnarray*}
We now show that it is possible to construct a cover
\begin{IEEEeqnarray*}{rCl}
	\label{eq: D'}
	\mathcal{D}'=\{\mathbf{s}'=(\mathbf{x}',\mathbf{y}'): \mathbf{x}'\in\mathcal{X}', \mathbf{y}'=\mathbf{A}\mathbf{x}' \in \mathcal{Y}'\} {\subseteq} \mathcal{X}' \times \mathcal{Y}'
\end{IEEEeqnarray*}
of
\begin{IEEEeqnarray*}{rCl}
	\label{eq: D}
	\mathcal{D}=\{\mathbf{s}=(\mathbf{x},\mathbf{y}): \mathbf{x}\in\mathcal{X}, \mathbf{y}=\mathbf{Ax}+\mathbf{n} \in \mathcal{Y}\} {\subseteq} \mathcal{X} \times \mathcal{Y}
\end{IEEEeqnarray*}
with metric balls of radius less than or equal (i) $\left(\delta^r+(\Lambda_{a}\delta+\eta)^r\right)^{\nicefrac{1}{r}}$ (for $r$-norm product metric); (ii) $\max(\delta,\Lambda_{a}\delta+\eta)$ (for sup product metric); and (iii) $(1+\Lambda_{a})\delta+\eta$ (for sum product metric). Concretely,
\begin{enumerate}
    \item[i -- ] For $r$-norm product metric $\rho=\left(\|\mathbf{y}-\mathbf{y}'\|^r_2+\|\mathbf{x}-\mathbf{x}'\|_2^r\right)^{\nicefrac{1}{r}}$, $\forall \mathbf{s} = (\mathbf{x},\mathbf{y})) \in \mathcal{D}, \exists \mathbf{s}'=(\mathbf{x}',\mathbf{y}'=\mathbf{A} \mathbf{x}')) \in \mathcal{D}'$ such that
    \begin{IEEEeqnarray*}{rCl}
		\rho(\mathbf{s},\mathbf{s}')&=&\left(\|\mathbf{x}-\mathbf{x}'\|^r_2+\|\mathbf{y}-\mathbf{y}'\|^r_2\right)^{\nicefrac{1}{r}}\\
		&{\leq} &\left(\delta^r+(\Lambda_{a}\delta+\eta)^r\right)^{\nicefrac{1}{r}}
	\end{IEEEeqnarray*}		
    \item[ii -- ]  For sup product metric, $\rho=\max(\|\mathbf{y}-\mathbf{y}'\|_2,\|\mathbf{x}-\mathbf{x}'\|_2)$, $\forall \mathbf{s} = (\mathbf{x},\mathbf{y})) \in \mathcal{D}, \exists \mathbf{s}'=(\mathbf{x}',\mathbf{y}'=\mathbf{A} \mathbf{x}')) \in \mathcal{D}'$ such that
      \begin{IEEEeqnarray*}{rCl}
		\rho(\mathbf{s},\mathbf{s}')&=&\max(\|\mathbf{x}-\mathbf{x}'\|_2,\|\mathbf{y}-\mathbf{y}'\|_2)\\
		&{\leq}&\max(\delta,\Lambda_{a}\delta+\eta)
	\end{IEEEeqnarray*}	
    \item[iii -- ] For sum product metric, $\rho=\|\mathbf{y}-\mathbf{y}'\|_2+\|\mathbf{x}-\mathbf{x}'\|_2$, $\forall \mathbf{s} = (\mathbf{x},\mathbf{y})) \in \mathcal{D}, \exists \mathbf{s}'=(\mathbf{x}',\mathbf{y}')) \in \mathcal{D}'$ such that 
	  \begin{IEEEeqnarray*}{rCl}
		\rho(\mathbf{s},\mathbf{s}')&=&\left(\|\mathbf{x}-\mathbf{x}'\|_2+\|\mathbf{y}-\mathbf{y}'\|_2\right)\\
		&{\leq}&(1+\Lambda_{a})\delta+\eta
	\end{IEEEeqnarray*}		
\end{enumerate}

Therefore, in presence of a $\Lambda_a$-lipschitz mapping, a $\delta$-cover $\mathcal{X}'$ of $\mathcal{X}$ induces a $\rho$-cover $\mathcal{D}'$ over the product space $\mathcal{D}$ such that the cardinality of the set $\mathcal{D}'$ is equal to the cardinality of the set $\mathcal{X}'$ -- thus proving the lemma.
\end{proof}	
We are now in a position to prove the Theorem. 
 We first note that
	\begin{IEEEeqnarray}{rCl}\label{eq : Lip_cont_reg_loss}
	\nonumber|l(f_\mathcal{S},\mathbf{s}_2)-l(f_\mathcal{S},\mathbf{s}_1)|
	&=&\big|\|\mathbf{x}_2-f_\mathcal{S}(\mathbf{y}_2)\|_2-\|\mathbf{x}_1-f_\mathcal{S}(\mathbf{y}_1)\|_2\big|\nonumber\\
	&\stackrel{(a)}{\leq}& \|\mathbf{x}_2-f_\mathcal{S}(\mathbf{y}_2)-\mathbf{x}_1+f_\mathcal{S}(\mathbf{y}_1)\|_2\nonumber\\
	&\stackrel{(b)}{\leq}& \|\mathbf{x}_2-\mathbf{x}_1\|_2+\|f_\mathcal{S}(\mathbf{y}_2)-f_\mathcal{S}(\mathbf{y}_1)\|_2\nonumber\\
	&\stackrel{(c)}{\leq}& \left(1+\Lambda_{f\circ a}\right)\|\mathbf{x}_2-\mathbf{x}_1\|_2+2\Lambda_f\eta
	\end{IEEEeqnarray}
	The inequalities $(a)$ and $(b)$ hold due to reverse triangle inequality and Minkowski-inequality. The inequality in $(c)$ is established due to Theorem \ref{thm: lipschitz_cont_end2end}. 

	{ We also note from Lemma 1 that we can partition the set $\mathcal{D}$ onto at most $K=\cN_\cX(\nicefrac{\delta}{2},\ell_2)$ (disjoint)  partitions such that if $(\mathbf{x}_1,\mathbf{y}_1) \in \mathcal{D}$ and $(\mathbf{x}_2,\mathbf{y}_2) \in \mathcal{D}$ are within the same partition then
	\begin{IEEEeqnarray}{rCl}
	\|\mathbf{x}_1-\mathbf{x}_2 \|_2 \leq \delta
	\end{IEEEeqnarray}
	hence
	\begin{IEEEeqnarray}{rCl}
	|l(f_\mathcal{S},\mathbf{s}_1))-l(f_\mathcal{S},\mathbf{s}_2)|&\le& \left(1+\Lambda_{f\circ a}\right)\delta+2\eta\Lambda_{f}
	\end{IEEEeqnarray}
	The Theorem then follows immediately from Definition 1.
	} 
	\end{proof}
	\begin{proof}[Proof of Theorem \ref{thm: ge}]

We first establish a simple Lemma.

\begin{lem}
\label{lem: bounded_Lip}
The Lipschitz constant of a differentiable function $f$ on a compact set $\mathcal{Z}$ is bounded.
\end{lem}
\begin{proof}
Let $f:\mathbb{R}^p\rightarrow\mathbb{R}^q$ be a differentiable function, defined on a compact set $\mathcal{Z}\subseteq\mathbb{R}^p$. Let also $g(\theta) = f(\mathbf{z}' + \theta(\mathbf{z}''- \mathbf{z}'))$, for some $\theta\in[0,1]$, so that $g(0)=f(\mathbf{z}')$ and $g(1)=f(\mathbf{z}'')$ where $\mathbf{z}',\mathbf{z}''$ are any two fixed points. Then, by the fundamental theorem of calculus, we have
\begin{IEEEeqnarray*}{rCl}
f(\mathbf{z}')-f(\mathbf{z}'') = {\int_{0}^{1}\mathbf{J}(\mathbf{z}' + \theta(\mathbf{z}''- \mathbf{z}'))d\theta}({\mathbf{z}'-\mathbf{z}''})
\end{IEEEeqnarray*}
where $\mathbf{J}(\mathbf{z})$ is the Jacobian matrix of $f$ at $\mathbf{z}$.

From the multiplicative property of norms, we also have that
\begin{IEEEeqnarray*}{rCl}
\|f(\mathbf{z}')-f(\mathbf{z}'')\|
&\le&\left\|{\int_{0}^{1}\mathbf{J}(\mathbf{z}' + \theta(\mathbf{z}''- \mathbf{z}'))d\theta}\right\|_2\|{\mathbf{z}'-\mathbf{z}''}\|_2
\end{IEEEeqnarray*}

Next, by the triangle inequality for integrals, it can be shown that
\begin{IEEEeqnarray*}{rCl}
\left\|{\int_{0}^{1}\mathbf{J}(\mathbf{z}' + \theta(\mathbf{z}''- \mathbf{z}'))d\theta}\right\|_2
&\le& \sup_{\stackrel{\mathbf{z}',\mathbf{z}''\in\mathcal{Z}} {\theta\in[0,1]}}\|\mathbf{J}(\mathbf{z}' + \theta(\mathbf{z}''- \mathbf{z}'))\|_2\\
&\le&\sup_{\mathbf{z}\in conv(\mathcal{Z})}\|\mathbf{J}(\mathbf{z})\|_2
\end{IEEEeqnarray*}
where $conv(\mathcal{Z})$ represents the convex hull of the compact set $\mathcal{Z}$. Note that the Carath\'eodory’s theorem of convex hulls can be used to prove that the convex hull of compact set in a finite dimensional space $\mathbb{R}^p$ is also compact \cite{gallier2008notes}.

Next, for a continuous function $f$ defined on a compact set, there exists a finite $\lambda_0$ such that
 \cite{bolzano1930functionenlehre, frechet1904generalisation}.
\begin{IEEEeqnarray}{rCl}
\label{eq: bounded_derivative}
\left|\frac{\partial}{d z_j}(f(\mathbf{z})_i) \right|\le \lambda_0
\end{IEEEeqnarray}
where $\frac{\partial}{d z_j}(f(\mathbf{z})_i)$ is the element at row $(i,j)$-th element of the Jacobian matrix $\mathbf{J}$. This, then leads to the following 
\begin{IEEEeqnarray*}{rCl}
\sup_{\mathbf{z}\in conv(\mathcal{Z})}\|\mathbf{J}(\mathbf{z})\|_2
\stackrel{(a)}{\le}\sup_{\mathbf{z}\in conv(\mathcal{Z})}c\|\mathbf{J}(\mathbf{z})\|_\infty\stackrel{(b)}{\le}cp\lambda_0
\end{IEEEeqnarray*}
where $(a)$ is due to the equivalence of matrix norms and $c$ is a constant dependent on the dimensions of the Jacobian matrix \cite{tonge2000equivalence}. Finally the last inequality follows form the definition of the $\|.\|_\infty$ matrix norm \cite{lewis2010top}. 
\end{proof}

We are now in a position to prove the Theorem. In particular, it can be shown that the $GE$ of a $(K,\epsilon(\mathcal{S}))$-robust deep neural network, with probability greater than $1-\zeta$, obeys \cite{xu2012robustness}
	\begin{IEEEeqnarray}{rCl}
	\label{eq: GE_Xu}
	GE
	& \le& \epsilon(\mathcal{S})+\max_\mathbf{s}|l(f_\mathcal{S},\mathbf{s})|\sqrt{\frac{2K\log(2)+2\log(1/\zeta)}{m}}
	\end{IEEEeqnarray}

We can immediately use the robustness result in Theorem 2 to determine two quantities this generalization error bound: $\epsilon(\mathcal{S})$ and $K$. However -- in contrast with existing results that assume that the loss function is uniformly bounded so that $\max_\mathbf{s}|l(f_\mathcal{S},\mathbf{s})|\le M<\infty$ (e.g. see \cite{xu2012robustness}) -- the loss function associated with our inverse problem is not necessarily bounded. However, it is still possible to show that $\max_\mathbf{s}|l(f_\mathcal{S},\mathbf{s})|$ is finite.

In particular, let us observe that $\forall~\mathbf{s}=(\mathbf{x},\mathbf{y}),\mathbf{s}'=(\mathbf{x}',\mathbf{y}')\in\mathcal{D}$ 
	\begin{IEEEeqnarray*}{rCl}
|l(f_\mathcal{S},\mathbf{s})-l(f_\mathcal{S},\mathbf{s}')|&=&
\big|\|\mathbf{x}-f_\mathcal{S}(\mathbf{y})\|_2-\|\mathbf{x}'-f_\mathcal{S}(\mathbf{y}')\|_2\big|\nonumber\\
	&{\leq}& \|\mathbf{x}-\mathbf{x}'\|_2+\|f_\mathcal{S}(\mathbf{y})-f_\mathcal{S}(\mathbf{y}')\|_2\\
&\stackrel{(a)}{\le}& \|\mathbf{x}-\mathbf{x}'\|_2+\Lambda_{f}\|\mathbf{y}-\mathbf{y}'\|\\
&\stackrel{(b)}{\le}& (1+\Lambda_{f})\rho(\mathbf{s},\mathbf{s}')
\end{IEEEeqnarray*}
where $(a)$ is due to Corollary 2 in \cite{sokolic2017robust} and $(b)$ holds because product space metric upper bounds the metrics on constituent metric spaces $\mathcal{X}$ and $\mathcal{Y}$ \footnote{This is true for most product metrics such as $\sup$, sum and $r$-norm product metric considered in Lemma \ref{lem: product_sp_covnumber}.}.

Let us also observe that -- due to Lemma \ref{lem: bounded_Lip} -- the Lipschitz constant of the loss function is finite because the Lipschitz constant of the neural network $\Lambda_f$ is also finite.

This immediately implies that the loss function is Lipschitz continuous hence continuous, and -- by the Extreme Value theorem \cite{frechet1904generalisation} -- that it is also bounded on $\cD$, so that $\max_\mathbf{s}|l(f_\mathcal{S},\mathbf{s})|\le M<\infty$.

The Theorem then follows immediately from Theorem 2.

\end{proof}

\section{Algorithms}
\subsection{Computation of the $jvp$}
The steps for efficiently computing the Jacobian vector product, $jvp$, are given in Algorithm \ref{alg: jvp} \cite{townsend2017}.
\begin{algorithm2e}[h]
\small
\KwIn{ Mini-batch $\cB$, model outputs $f(\mathbf{y})$, vector $\mathbf{v}$.}
\KwOut{$\mathbf{Jv}$}
        Initialize a dummy tensor $\mathbf{d}$.\\
        $\mathbf{g}\leftarrow vjp(f(\mathbf{y}), \mathbf{y}, \mathbf{d})$\\
        $\mathbf{u}_1\leftarrow vjp(\mathbf{g}, \mathbf{p}, \mathbf{v})$\\
        \Return $\mathbf{u}$
    \caption{
      	Computation of the $jvp$.
    }
    \label{alg: jvp}
\end{algorithm2e}
\subsection{Square of the Frobenius norm of Jacobian terms}
In \cite{hoffman2019robust}, random projections have been leveraged to efficiently compute the square of the Frobenius norm of the Jacobian matrix of a neural network. This algorithm can be immediately specialized to approximate the square of the Frobenius norm of $\mathbf{JA}$ too.

We report for the sake of completeness the original algorithm in Alg. \ref{alg: frob_j}. 
\begin{algorithm2e}[ht]
\small
\KwIn{ Mini-batch $\cB$, number of projections $n$.}
\KwOut{Square of the Frobenius norm of the Jacobian $\mathcal{J}_F$.}
$\mathcal{J}_F \leftarrow 0$\\
\For {$(\mathbf{y},\mathbf{x})\in\cB$}
{
$i \leftarrow 0$\\
\While{ $i < n$ }
{ Initialize $\left\{\mathbf{z}\right\} \sim \mathcal{N}(0, \mathbf{I})$\\ 
${\mathbf{z}} \leftarrow \mathbf{z}/\vert\vert\mathbf{z}\vert\vert$ \\
$\mathcal{J}_F \leftarrow \mathcal{J}_F\mathrel{+} p \|vjp(f(\mathbf{y}),\mathbf{y},\mathbf{z})\|_2^2 / (n|\cB|)$
} }
    \caption{
            Estimation of the $\|\mathbf{J}\|_F^2$
    }
    \label{alg: frob_j}
\end{algorithm2e}
We now report the modified algorithm to compute $\|\mathbf{JA}\|_F^2$ in Alg. \ref{alg: frob_ja}.
\begin{algorithm2e}[ht]
\small
\KwIn{ Mini-batch $\cB$,number of projections $n$.}
\KwOut{Square of the Frobenius norm of the Jacobian $\mathcal{A}_F$.}
$\mathcal{A}_F \leftarrow 0$\\
\For {$(\mathbf{y},\mathbf{x})\in\cB$}
{
$i \leftarrow 0$\\
\While{ $i < n$ }
{ Initialize $\left\{\mathbf{z}\right\} \sim \mathcal{N}(0, \mathbf{I})$\\ 
${\mathbf{z}} \leftarrow \mathbf{z}/\vert\vert\mathbf{z}\vert\vert$ \\
$\mathcal{A}_F \leftarrow \mathcal{A}_F\mathrel{+} p \|vjp(f(\mathbf{y}),\mathbf{y},\mathbf{z})\cdot \mathbf{A}\|_2^2 / (n|\cB|)$
} }
    \caption{
            Estimation of the $\|\mathbf{JA}\|_F^2$
    }
    \label{alg: frob_ja}
\end{algorithm2e}

\subsection{Spectral norm of $\mathbf{J}$ and \texorpdfstring{$\mathbf{JA}$}{JA}}
The objective function specified in Section \ref{sec: discussion}, eq. (\ref{eq: loss_spectral}) of the main manuscript required spectral norm of both $\mathbf{J}$ and $\mathbf{JA}$ to compute the regularized loss. The algorithm for the calculation of the spectral norm of Jacobian has been included in the main manuscript. The spectral norm of Jacobian times the forward map $\mathbf{A}$ can be computed easily by modifying Algorithm \ref{alg:  spectral_j} as shown in Algorithm \ref{alg: spectral_ja}.

\begin{algorithm2e}[ht]
\small
\KwIn{ Mini-batch $\cB$,number of power iterations $n$.}
\KwOut{Maximum singular value $\sigma$ of the $\mathbf{JA}$ matrix.}
\For {$(\mathbf{y},\mathbf{x})\in\cB$}
{
    Initialize $\left\{\mathbf{v}\right\} \sim \mathcal{N}(0, \mathbf{I})$\\
    $i \leftarrow 0$\\
    \While{ $i < n$ }
    {
        $\mathbf{u} \leftarrow jvp(f(\mathbf{y}),\mathbf{y},\mathbf{Av})$\\
        $\mathbf{v} \leftarrow vjp(f(\mathbf{y}), \mathbf{y}, \mathbf{A}^T\mathbf{u})$ \\
         
        $i \leftarrow i+1$.
    }$\sigma\leftarrow\|\mathbf{u}\|_2/\|\mathbf{v}\|_2$} 
    \caption{
            Esitmation of the spectral norm of $\mathbf{JA}$
    }
    \label{alg: spectral_ja}
\end{algorithm2e}
\subsubsection{Accuracy \& Efficiency} \label{sec: acc_jvp}
\begin{table}[h]
\scriptsize
    \centering
    \caption{\small
        Time and memory requirements for training a 4-layer fully connected NN and 5-layer DnCNN \cite{zhang2017beyond} on the full training set of Fashion MNIST with a batch size of 100 and $p=q=784$. 
    }
    \vspace{0.5em}
    \begin{tabular}{r c rr c rr}
        \toprule
                && \multicolumn{2}{c}{4-layer FC NN} & \multicolumn{2}{c}{5-layer DnCNN }\\
                \cmidrule{3-4} \cmidrule{5-6}
                && time & memory & time & memory\\
        \midrule
        Vanilla &&  $29$m &   $595$Mb &    $1$h$8$m &   $1057$Mb\\
        \cite{hoffman2019robust} &&  $2$h$23$m &   $659$Mb &    $6$h,$10$m &   $1363$Mb\\
        Alg. \ref{alg:  spectral_j} ($n=1$) &&  $47.5$m &   $659$Mb &    $3$h,$42$m &   $1825$Mb\\
    Alg. \ref{alg:  spectral_j} ($n=2$) &&   $1$h,$1$m &    $787$Mb &    $5$h,$4$m &    $2849$Mb\\
        Alg. \ref{alg:  spectral_j} ($n=3$) &&   $1$h,$13$m &    $787$Mb &      $6$h,$31$m &    $4897$Mb\\
        Alg. \ref{alg:  spectral_j} ($n=4$) &&   $1$h,$18$m &    $787$Mb &      $9$h,$42$m &    $4897$Mb\\
        tf batch J ($n=3$) &&   $63$h,$7$m &    $4659$Mb &    $--$ &    $\approx160$Gb \\
        \bottomrule
    \end{tabular}
    \label{tab: jac_compute}
\end{table}
\begin{figure}[h]
    \centering
    \includegraphics[width=0.65\linewidth]{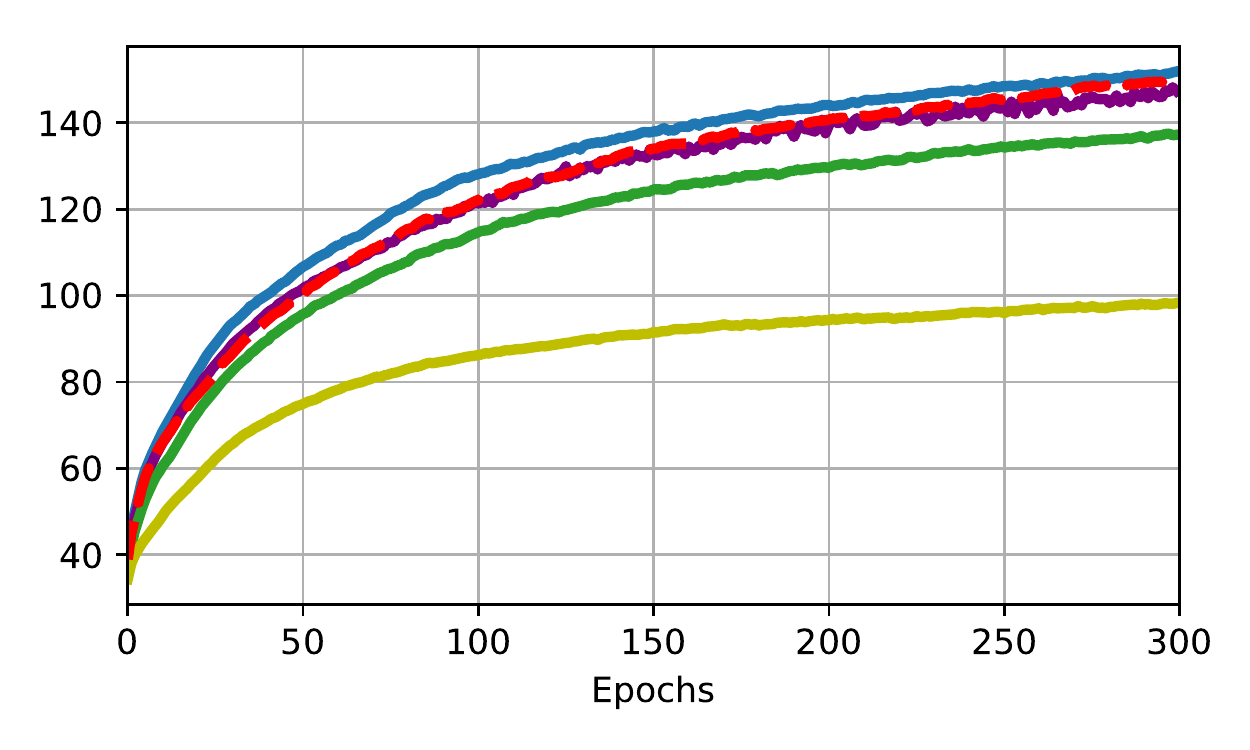}\\
    \includegraphics[width=1.\linewidth]{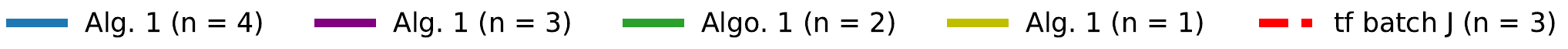}
    \caption{    Maximum singular values of the batch Jacobians for a 4-layer fully connected network with $p=q=784$.
    }
    \vspace{-1em}
    \label{fig:singular-values}
\end{figure}
We validate that Algo. \ref{alg:  spectral_j} indeed results in a faithful estimate of the spectral norm of the Jacobian. To do so, we compare the output of our algorithm with the output of the power method applied to a Jacobian matrix computed using Tensorflow while training a fully connected network. We plot the resulting outputs of both methods as a function of the number of epochs in Fig \ref{fig:singular-values}. It can be seen that for equal number of power iterations $(n=3)$, the results obtained using both methods are almost identical.

Next,we look at the computational resources required for regularizing different models using Alg. \ref{alg:  spectral_j} and compare to the resource consumption of other Jacobian based regularization algorithms. 
In Table \ref{tab: jac_compute}, we offer a comparative analysis in terms of time and memory requirement of training different models on MNIST using Adam with and without explicit Jacobian regularization schemes for a batch size, $|\mathcal{B}|=100$. We use a Tesla-V100 16 GB for this test. 

For both fullly connected and convolutional neural networks, our analysis shows that regularizing the network using Alg. \ref{alg:  spectral_j}, in contrast to regularizing with Jacobian regularization proposed in \cite{hoffman2019robust} offers gains of orders of magnitudes in terms of computation requirements. On the other hand regularizing a neural network by first computing the Jacobian and then calculating the spectral norm is practically impossible even for a modestly sized fully connected network. For convolutional neural networks, even a minbatch Jacobian of $10$ samples occupies $16$GB of memory making it infeasible to computes $\|\mathbf{J}\|_2$ at all.

In its current state, Alg. \ref{alg:  spectral_j} requires one forward pass to compute $jvp$ and two forward passes to compute $vjp$. This computational overhead can be reduced if support for forward mode jacobian vector products is incorporated in modern machine learning frameworks. Apart from the application discussed in this work, the $jvp$s play a pivotal role in a number of other techniques such as calculating Hessian vector product, which in turn can be used for second order optimization \cite{martens2010deep}. New frameworks with the specialized support for advanced automatic differentiation have been introduced that may result in reduced costs associated with our algorithm but they are still in the research and development stage \cite{jax2018github}.

{ Note that the for loop in Alg. \ref{alg:  spectral_j}, \ref{alg: frob_j}, \ref{alg: frob_ja} and \ref{alg:  spectral_ja} is only for expositional purposes. All the modern machine learning frameworks can implement these algorithms over the whole batch in a parallel fashion.}

\section{Experiments}\label{sec: experiments_apndx}
\subsection{Gaussian Measurements}
For Gaussian measurements, we used a $4$-layer fully connected neural network for reconstruction. We fine tuned the regularization parameters $\lambda_1$, $\lambda_2$ for our regularizers and all the competing methods (WD, WS, FJ) using a grid search.

Additional reconstruction results from the Gaussian measurements of the ground truth image in Figure \ref{fig: gaussian_gt_3} using varius regularizers is given in Figure \ref{fig: gaussian_recon_3}.
\subsection{k-space subsampled measurements}
For $k$-space subsampled measurements, we used the UNet architecture  shown in Figure \ref{fig: unet} \cite{ronneberger2015u}. 

For this set of experiments, we compared the reconstruction performance of the proposed Jacobian regularization methods with the adversarial regularizer \cite{lunz2018adversarial} and postprocessing via UNet method \cite{jin2017deep}.

Both the postprocessing and the adversarial regularization method involve a `preprocessing' step. That is, both techniques require a classical regularized reconstruction method, $\mathbf{A}^\dagger(.)$, which encapsulates the knowledge of the forward model to be applied to the measurement $\mathbf{y}$. This $\mathbf{A}^\dagger(\mathbf{y})$ is then used as input to these reconstruction algorithms. For our experimental setting, we applied the $\ell_1$ wavelet regularized reconstruction method to the subsampled measurements.

We modified the official implementation of the adversarial regularizer, present on Github, provided by the authors of the publication to suit the forward model used in this work. The batch size and other hyperparameters such as the step size and the choice of the adversarial regularizer network were kept the same as in the original implementation. The authors provided a closed form technique to compute the hyperparamter $\lambda$ used in their algorithm. We therefore used the sugggested technique and performed no further fine-tuning.   

For the proposed Jacobian regularization method, an adaptive policy -- which took the feedback from training into account -- was used to tune the hyperparamters $\lambda_1$ and $\lambda_2$ as opposed to keeping them fixed. Our empirical results show that using such an adaptive technique results in better validation performance. Since this technique takes into account the training performance to compute the regularization coefficients at each step and does not rely on a hit and trial method to find the `best' hyperparamter, it results in lesser overall training time for the model. The steps for computing the adaptive regularization parameter are summarized in Algorithm \ref{alg: adaptive_lambda}.
\begin{algorithm2e}[ht]
\small
\KwIn{magnitude $r$ of the regularization term and $l$ of the loss over Mini-batch $\cB$, scaling factor $s$ }
\KwOut{Value of the regularization coefficient}
$\alpha \leftarrow \text{floor}(\log(\nicefrac{l}{r}))$ \tcp*{$l$ is the unregularized empirical loss $\nicefrac{1}{|\mathcal{B}|}\sum_il(f_\cS,\mathbf{s}_i)$}
$\lambda \leftarrow \nicefrac{\alpha}{s}$ \tcp*{The values of $10,20$ and $30$ were tested for $s$. $20$ usually gave the best results.}

    \caption{
            Estimation of the regularization coefficient $\lambda$ for Jacobian regularizer.
    }
    \label{alg: adaptive_lambda}
\end{algorithm2e}
Figures \ref{fig: mri_recon2Dmask} and \ref{fig: mri_recon1Dmask} present additional results for the MRI experiments for the acquisition masks in Figure \ref{fig: 2DMask} and \ref{fig: 1Dmask}. Since the Jacobian regularization method can be directly used with any deep learning based reconstruction method, we also include resconstruction results when a postprocessing UNet is regularized via SJA\&SJ regularizer. Perceptually the reconstruction achieved through this method outperforms all the other techniques. However there is no improvement in terms of PSNR and SSIM over the UNet with SJA\&SJ (without the preprocessing).
\begin{figure}[h]
    \centering
      \includegraphics[width=0.15\textwidth]{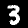}   \\
    \vspace{0.2em}
    \caption{
       Grount Truth image corresponding to the reconstruction in Figure \ref{fig: gaussian_recon_3}. 
    }
    \label{fig: gaussian_gt_3}
    \vspace{-1em}
\end{figure}
\begin{figure}[h]
  \centering
        \includegraphics[width=0.15\textwidth]{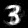} 
      \includegraphics[width=0.15\textwidth]{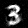}
      \includegraphics[width=0.15\textwidth]{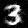}
      \includegraphics[width=0.15\textwidth]{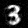}
      \includegraphics[width=0.15\textwidth]{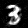} 
    \\
    \vspace{0.1em}    
     \includegraphics[width=0.15\textwidth]{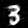} 
     \includegraphics[width=0.15\textwidth]{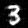}
     \includegraphics[width=0.15\textwidth]{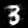}
     \includegraphics[width=0.15\textwidth]{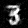}
     \includegraphics[width=0.15\textwidth]{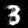} 
    \\
    \vspace{0.1em}    
     \vspace{0.1em}    
     \includegraphics[width=0.15\textwidth]{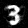} 
     \includegraphics[width=0.15\textwidth]{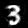}
     \includegraphics[width=0.15\textwidth]{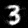}
     \includegraphics[width=0.15\textwidth]{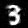}
     \includegraphics[width=0.15\textwidth]{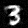} 
      \\
    \vspace{0.1em}    
     \vspace{0.1em}    
     \includegraphics[width=0.15\textwidth]{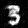} 
     \includegraphics[width=0.15\textwidth]{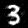}
     \includegraphics[width=0.15\textwidth]{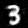}
     \includegraphics[width=0.15\textwidth]{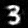}
     \includegraphics[width=0.15\textwidth]{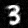} 
    \\
    \vspace{0.2em}
    \caption{
        Sample results from the reconstruction from compressed Guassian measurements for the ground truth image in Figure \ref{fig: gaussian_gt_3}. (Top to Bottom) Reconstruction from $80,160,320$ and $640$ measuremenst. (Left to Right) reconstruction using a 4-layer FC neural network regularized with WD, SW, FJ, FJA\&FJ and SJA\&SJ. 
    }
    \label{fig: gaussian_recon_3}
    \vspace{-1em}
\end{figure}
\begin{figure}[t]
    \centering
    \includegraphics[width=0.75\linewidth]{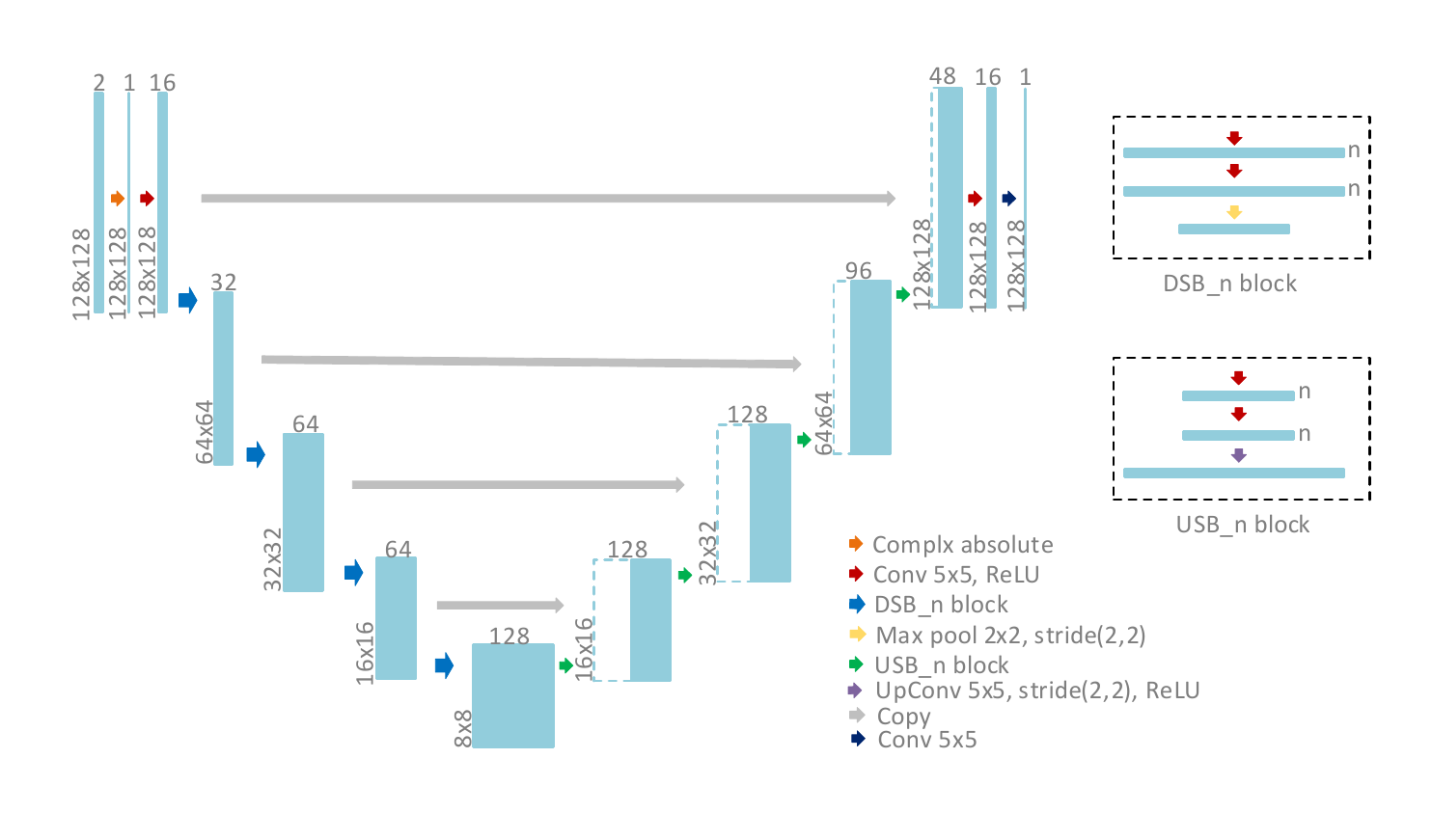}
    \caption{UNet used in our experiments.
    }
    \vspace{-1em}
    \label{fig: unet}
\end{figure}
\begin{figure}[t]
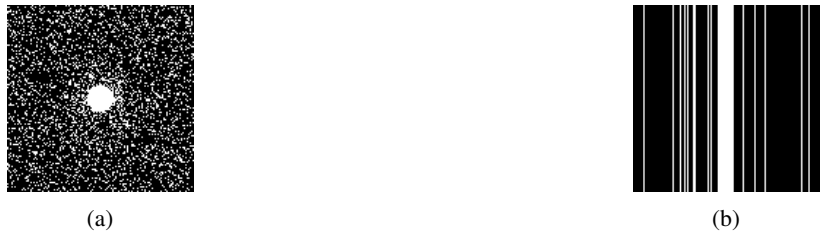

\begin{subfigure}{0.50\textwidth}
    \centering
    \includegraphics[width=0.30\linewidth]{k_space_results/circ_centre_fs/mask_s128d5r5.png}
    \caption{
    }
\label{fig: 2DMask}
\end{subfigure}
\begin{subfigure}{0.5\textwidth}
  \centering
      \includegraphics[width=0.30\textwidth]{k_space_results/rect_centre_fs/mask_s128cf8acc4.png}
    \caption{ }
\label{fig: 1Dmask}
\end{subfigure}
\caption{k-space acquisition masks (a): Random 2D 5-fold subsampling mask with the centre fully sampled.(b) Random 1D 4-fold subsampling mask.
}
\label{fig:masks}
\end{figure}
\begin{figure}[h]
    \centering
      \includegraphics[width=0.13\textwidth]{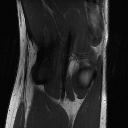} 
    \includegraphics[width=0.13\textwidth]{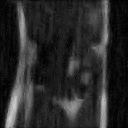}
     \includegraphics[width=0.13\textwidth]{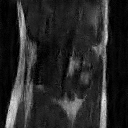} 
    \includegraphics[width=0.13\textwidth]{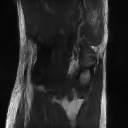} 
    \includegraphics[width=0.13\textwidth]{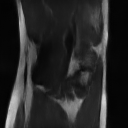}
      \includegraphics[width=0.13\textwidth]{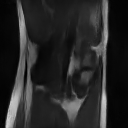}
    \includegraphics[width=0.13\textwidth]{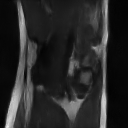}
    \\
    \vspace{0.1em}    
    \includegraphics[width=0.13\textwidth]{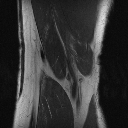} 
    \includegraphics[width=0.13\textwidth]{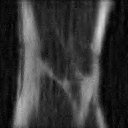}
     \includegraphics[width=0.13\textwidth]{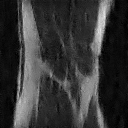} 
    \includegraphics[width=0.13\textwidth]{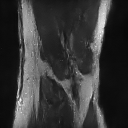} 
    \includegraphics[width=0.13\textwidth]{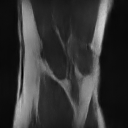}
      \includegraphics[width=0.13\textwidth]{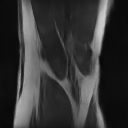}
    \includegraphics[width=0.13\textwidth]{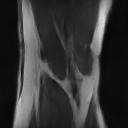}
    \\
    \vspace{0.1em}
    \includegraphics[width=0.13\textwidth]{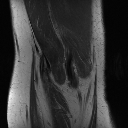} 
    \includegraphics[width=0.13\textwidth]{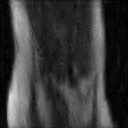}
    \includegraphics[width=0.13\textwidth]{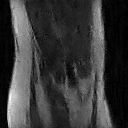} 
    \includegraphics[width=0.13\textwidth]{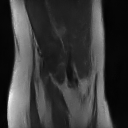} 
    \includegraphics[width=0.13\textwidth]{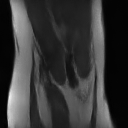}
    \includegraphics[width=0.13\textwidth]{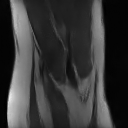}
    \includegraphics[width=0.13\textwidth]{additional_results/2DMask/ppSJASJ_predi_000125.png}
      \\
    \vspace{0.1em}    
    \includegraphics[width=0.13\textwidth]{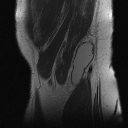} 
    \includegraphics[width=0.13\textwidth]{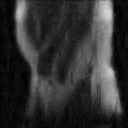}
     \includegraphics[width=0.13\textwidth]{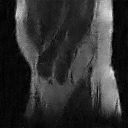} 
    \includegraphics[width=0.13\textwidth]{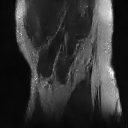} 
    \includegraphics[width=0.13\textwidth]{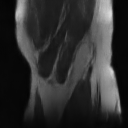}
      \includegraphics[width=0.13\textwidth]{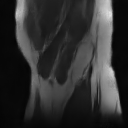}
    \includegraphics[width=0.13\textwidth]{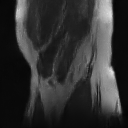}
    \\
    \vspace{0.2em}
    \caption{
        Sample results from the reconstruction from k-space subsampled measurements acquired using a 2D acquisition mask in Figure \ref{fig: 2DMask}. (Left to Right) ground truth, reconstruction using $\ell_1$ Wavelet regularization, Adversarial Regularizer \cite{lunz2018adversarial}, postprocessing using UNet \cite{jin2017deep}, UNet with FJA\&FJ and UNet with SJA\&SJ, postprocessing using UNet \cite{jin2017deep} regualrized with SJA\&SJ. 
    }
    \label{fig: mri_recon2Dmask}
    \vspace{-1em}
\end{figure}
\begin{figure}[h]
    \centering
      \includegraphics[width=0.13\textwidth]{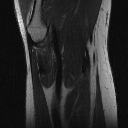} 
    \includegraphics[width=0.13\textwidth]{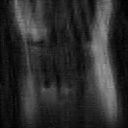}
     \includegraphics[width=0.13\textwidth]{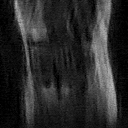} 
    \includegraphics[width=0.13\textwidth]{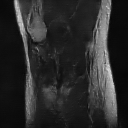} 
    \includegraphics[width=0.13\textwidth]{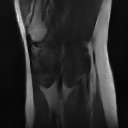}
      \includegraphics[width=0.13\textwidth]{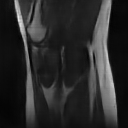}
    \includegraphics[width=0.13\textwidth]{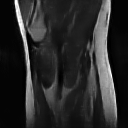}
       \\
    \vspace{0.1em}
     \includegraphics[width=0.13\textwidth]{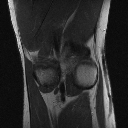} 
    \includegraphics[width=0.13\textwidth]{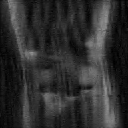}
     \includegraphics[width=0.13\textwidth]{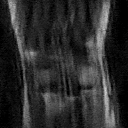} 
    \includegraphics[width=0.13\textwidth]{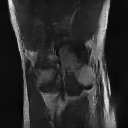} 
    \includegraphics[width=0.13\textwidth]{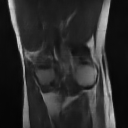}
      \includegraphics[width=0.13\textwidth]{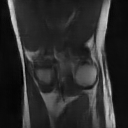}
    \includegraphics[width=0.13\textwidth]{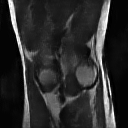}
       \\
    \vspace{0.1em}
     \includegraphics[width=0.13\textwidth]{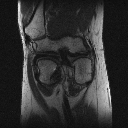} 
    \includegraphics[width=0.13\textwidth]{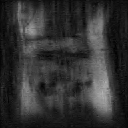}
     \includegraphics[width=0.13\textwidth]{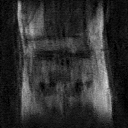}
    \includegraphics[width=0.13\textwidth]{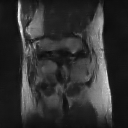}
    \includegraphics[width=0.13\textwidth]{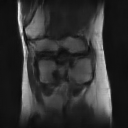}
      \includegraphics[width=0.13\textwidth]{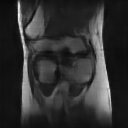}
    \includegraphics[width=0.13\textwidth]{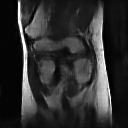}
     \\
    \vspace{0.1em}
     \includegraphics[width=0.13\textwidth]{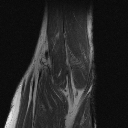} 
    \includegraphics[width=0.13\textwidth]{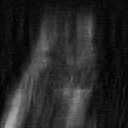}
     \includegraphics[width=0.13\textwidth]{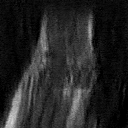}
    \includegraphics[width=0.13\textwidth]{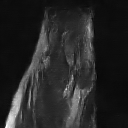}
    \includegraphics[width=0.13\textwidth]{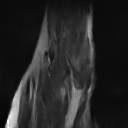}
      \includegraphics[width=0.13\textwidth]{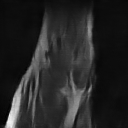}
    \includegraphics[width=0.13\textwidth]{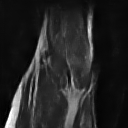}
    \\
    \vspace{0.2em}
    \caption{
        Sample results from the reconstruction from k-space subsampled measurements acquired using a 2D acquisition mask in Figure \ref{fig: 1Dmask}. (Left to Right) ground truth, reconstruction using $\ell_1$ Wavelet regularization, Adversarial Regularizer \cite{lunz2018adversarial}, postprocessing using UNet \cite{jin2017deep}, UNet with FJA\&FJ and UNet with SJA\&SJ, postprocessing using UNet \cite{jin2017deep} regualrized with SJA\&SJ. 
    }
    \label{fig: mri_recon1Dmask}
    \vspace{-1em}
\end{figure}
\end{document}